\documentclass[a4paper,12pt]{amsart}
\usepackage[lmargin=3.7cm,rmargin=2.7cm,bmargin=3.5cm,tmargin=3cm]{geometry}

\usepackage[latin1]{inputenc}
\usepackage{amsmath}
\usepackage{amssymb}  
\usepackage{amsthm}
\usepackage{booktabs}
\usepackage{textcomp}
\usepackage{enumerate}
\usepackage[sort,numbers]{natbib}
\usepackage{array}
\usepackage{psfrag}
\usepackage{graphicx}
\usepackage{bbm}

\newcommand{\trace}{\mathop{\mathrm{trace}}}
\newcommand{\defeq}{\mathrel{\mathop:}=} 
\newcommand{\eqdef}{=\mathrel{\mathop:}}

\newcommand{\ud}{\mathrm{d}}
\newcommand{\R}{\mathbb{R}}
\renewcommand{\P}{\mathbb{P}}

\newcommand{\spc}[1]{\mathbb{#1}}
\newcommand{\charfun}[1]{\mathbbm{1}_{#1}}
\newcommand{\Cov}{\mathrm{Cov}}
\newcommand{\diag}{\mathop{\mathrm{diag}}}

\newtheorem{theorem}{Theorem}

\newtheorem{proposition}[theorem]{Proposition}

\theoremstyle{definition}

\newtheorem{assumption}[theorem]{Assumption}

\theoremstyle{remark}
\newtheorem{remark}[theorem]{Remark}

\date{\today}
\begin{document}

\title[Robust Adaptive Metropolis]{Robust adaptive
  Metropolis algorithm with coerced acceptance rate}
\author{Matti Vihola}
\address{Matti Vihola, Department of Mathematics and Statistics,
  University of Jyväskylä,
  P.O.Box 35,
  FI-40014 University of Jyväskylä,
  Finland}
\email{matti.vihola@iki.fi}
\urladdr{http://iki.fi/mvihola/}
\thanks{The author was supported 
  by the Finnish Centre of Excellence in Analysis and
  Dynamics Research
  and by the Finnish Graduate School in
  Stochastics and Statistics.}

\maketitle

\begin{abstract} 
  The adaptive Metropolis (AM) algorithm of Haario, Saksman and Tamminen
  [Bernoulli 7 (2001) 223-242] uses the estimated covariance of the target
  distribution in the proposal distribution. This paper introduces a new
  robust adaptive Metropolis algorithm estimating the shape of the target
  distribution and simultaneously coercing the acceptance rate.  The
  adaptation rule is computationally simple adding no extra cost compared
  with the AM algorithm. The adaptation strategy can be seen as a
  multidimensional extension of the previously proposed method adapting the
  scale of the proposal distribution in order to attain a given acceptance
  rate. The empirical results show promising behaviour of the new algorithm
  in an example with Student target distribution having no finite second
  moment, where the AM covariance estimate is unstable. In the
  examples with finite second moments, the performance of the new approach
  seems to be competitive with the AM algorithm combined with scale
  adaptation.
\end{abstract} 

\section{Introduction} 
\label{sec:intro} 

Markov chain Monte Carlo (MCMC) is
a general method to approximate integrals of the form
\[
    I \defeq \int_{\R^d} f(x) \pi(x) \ud x < \infty
\]
where $\pi$ is a probability density function, which can be evaluated
point-wise up to a normalising constant. Such an integral occurs
frequently when computing Bayesian posterior expectations 
\cite[e.g.,][]{robert-casella,gilks-mcmc,roberts-rosenthal-general}. 
The MCMC method is based on a
Markov chain $(X_n)_{n\ge 1}$ that is easy to simulate in
practice, and for which the ergodic averages $I_n \defeq n^{-1}\sum_{k=1}^n
f(X_k)$ converge to the integral $I$ as the number of samples $n$ tends to
infinity.

One of the most generally applicable MCMC method is
the random walk Metropolis (RWM) algorithm. Suppose $q$ is a 
symmetric probability density supported on $\R^d$
(for example the standard Gaussian density) and let $S\in\R^{d\times d}$ be 
a non-singular matrix. Set $X_1 \equiv x_1$, where $x_1\in\R^d$ is a 
given starting point in the support; $\pi(x_1)>0$. 
For $n\ge 2$ apply recursively the following two steps:
\begin{enumerate}[(M1)]
\item simulate $Y_n = X_{n-1} + S U_n$, where $U_n\sim q$ is a independent
      random vector, and
      \label{item:metropolis-proposal}
\item with probability $\alpha_n \defeq \alpha(X_{n-1},Y_n) \defeq \min\{1,\pi(Y_n)/\pi(X_{n-1})\}$ 
      the proposal is accepted, and $X_n = Y_n$; otherwise 
      the proposal is rejected 
      and $X_n = X_{n-1}$.
\end{enumerate}
This algorithm will produce a valid chain, that
is, $I_n\to I$ almost surely as $n\to\infty$ \cite[e.g.][Theorem
1]{nummelin-mcmcs}. However, the efficiency of the method, that is, the
speed of the convergence $I_n\to I$, is crucially affected by the choice of
the shape matrix $S$.

Recently, there has been an increasing interest on adaptive MCMC algorithms
that try to learn some properties of the target distribution $\pi$
on-the-fly, and use this information to facilitate more efficient sampling
\citep{saksman-am,andrieu-robert,atchade-rosenthal,andrieu-moulines,%
roberts-rosenthal,roberts-rosenthal-examples}; see also the recent review
by \cite{andrieu-thoms}. In the context of the RWM algorithm, this is typically
implemented by replacing the constant shape $S$ in 
(M\ref{item:metropolis-proposal}) with a random
matrix $S_{n-1}$ that depends on the past (on the random
variables $U_k$, $X_k$, and $Y_k$ for $1\le k\le n-1$). 

Different strategies have been proposed to compute the matrix $S_{n-1}$. The
seminal Adaptive Metropolis (AM) algorithm \citep{saksman-am} uses $S_{n-1} =
\theta L_{n-1}$ where $L_{n-1}$ is the Cholesky factor of the (possibly
modified) empirical covariance matrix 
$C_{n-1}=\Cov(X_1,\ldots,X_{n-1})$. Under certain assumptions, the
empirical covariance converges to the true covariance of the target
distribution $\pi$ \cite[see, e.g.,
][]{saksman-am,andrieu-moulines,saksman-vihola,vihola-collapse}. The
constant scaling parameter $\theta>0$ is a tuning parameter chosen by the
user; the value $\theta = 2.4/\sqrt{d}$ proposed in the original paper is
widely used, as it is asymptotically optimal under certain theoretical
setting \citep{gelman-roberts-gilks}. 

In fact, the theory behind the value $\theta = 2.4/\sqrt{d}$ connects
the mean acceptance rate to the efficiency of the Metropolis algorithm
in more general settings. Therefore, it is sensible to try to find
such a scaling factor $\theta$ that yields a desired mean acceptance rate;
typically 23.4\% in multidimensional
settings~\citep{roberts-gelman-gilks-scaling}. 
The first algorithms coercing the acceptance rate
did not adapt the shape factor at all, but only the scale of the
proposal distribution. That is, $S_{n-1} = \theta_{n-1} I$, a multiple
of a constant matrix, where the factor $\theta_{n-1}\in(0,\infty)$ is
adapted roughly by increasing the value of the acceptance probability
is too low, and vice versa
\citep{atchade-rosenthal,andrieu-thoms,roberts-rosenthal-examples,atchade-fort}.
This adaptive scaling Metropolis (ASM) algorithm has some nice
properties, and it has been shown that the
algorithm is stable under quite a general setting \citep{vihola-asm}. 
It is, however, a
`one-dimensional' scheme, in the sense that it is unable to adapt to
the shape of the target distribution like the AM algorithm.  This can
result in slow mixing with certain target distributions $\pi$ having a
strong correlation structure.

The scale adaptation in the ASM approach has been proposed to be used 
within the AM algorithm~\citep{atchade-fort,andrieu-thoms}. 
This algorithm, which shall be referred here to as the adaptive
scaling within AM (ASWAM), combines the shape adaptation of AM and the acceptance
probability optimisation. Namely, $S_{n-1} = \theta_{n-1} L_{n-1}$, where
$\theta_{n-1}$ is computed from the observed acceptance probabilities
$\alpha_2,\ldots,\alpha_{n-1}$ and $L_{n-1}$ is
the Cholesky factor of $\Cov(X_1,\ldots,X_{n-1})$. This multi-criteria adaptation
framework provides a coerced acceptance probability, and at the same time
captures the covariance shape information of $\pi$. Empirical findings 
indicate this algorithm can overcome some difficulties
encountered with the AM method \citep{andrieu-thoms}.

The present paper introduces a new algorithm alternative to the ASWAM
approach. The aim is to seek a matrix factor $S_*$ that captures the shape 
of $\pi$ and at the same time allows to attain a given mean acceptance rate. Unlike
the multi-criteria adaptation in ASWAM, the new approach is based on a single
matrix update formula that is computationally equivalent to the covariance
factor update in AM. The algorithm, called here the robust adaptive 
Metropolis (RAM), differs from the ASWAM approach by avoiding the use
of the
empirical covariance, which can be problematic in some settings,
especially if $\pi$ has no finite second moment. The proposed approach
is reminiscent, yet not equivalent, with robust pseudo-covariance 
estimation, which has also been
proposed to be used in place of the AM approach \citep{andrieu-thoms}.

The RAM algorithm is described in detail in the next section. Section
\ref{sec:convergence} provides analysis on the stable points of the
adaptation rule, that is, where the sequence of matrices $S_n$ is
supposed to converge. In Section \ref{sec:lln}, the validity of the
algorithm is verified under certain sufficient conditions. It is also 
shown that the adaptation converges
to a shape of an elliptically symmetric target distribution.
The RAM algorithm was empirically tested in
some example settings and compared with the AM and the ASWAM 
approaches. Section \ref{sec:tests} summarises the encouraging findings. 
The final section concludes with some discussion on the approach
as well as directions of further research.


\section{Algorithm}
\label{sec:ashm} 

In what follows, 
suppose that the proposal density $q$ is spherically symmetric:
there exists a function $\hat{q}:\R\to[0,\infty)$ such that
$q(x) = \hat{q}(\|x\|)$ for all $x\in\R^d$.
Let $s_1\in\R^{d\times d}$ be a lower-diagonal matrix with positive diagonal
elements, and suppose 
$\{\eta_n\}_{n\ge 1}\subset(0,1]$ is a step size sequence decaying to zero.
Furthermore, let $x_1\in\R^d$ be some point in the support of the target
distribution, $\pi(x_1)>0$, and let $\alpha_*\in(0,1)$ stand for the 
target mean acceptance probability of the algorithm.

The robust adaptive Metropolis process is defined recursively through
\begin{enumerate}[({R}1)]
\item \label{step:proposal}
  compute $Y_n \defeq X_{n-1} + S_{n-1} U_n$, where $U_n\sim q$ is an 
      independent random vector,
\item \label{step:acc-rej} with probability $\alpha_n \defeq \min\{1,\pi(Y_n)/\pi(X_{n-1})\}$ 
      the proposal is accepted, and $X_n \defeq Y_n$; otherwise 
      the proposal is rejected 
      and $X_n \defeq X_{n-1}$, and
\item \label{step:adapt}
  compute the lower-diagonal matrix $S_n$ with positive diagonal elements
  satisfying the equation
  \begin{equation} 
\!    S_{n} S_{n}^T \!= S_{n-1}
    \left(I + \eta_{n}(\alpha_n\! - \alpha_*)\frac{U_{n}U_{n}^T}{\|U_{n}\|^2}\right)
        S_{n-1}^T
   \label{eq:ram-update}
   \end{equation}
   where $I\in\R^{d\times d}$ stands for the identity matrix.
\end{enumerate}
The steps (R\ref{step:proposal}) and (R\ref{step:acc-rej}) implement one
iteration of the RWM algorithm, but with a random matrix $S_{n-1}$ in
(R\ref{step:proposal}).
In the adaptation step (R\ref{step:adapt})
the unique $S_n$ satisfying \eqref{eq:ram-update} 
always exists, since it is the Cholesky factor of the
matrix in the right hand side, which is verified below to be 
symmetric and positive definite.
\begin{proposition} 
\label{prop:posdef} 
Suppose $S\in\R^{d\times d}$ is a non-singular matrix,
$u\in\R^d$ is a non-zero vector and $a\in(-1,\infty)$
is a scalar. Then, the matrix
$
M\defeq    S \big(I + a \frac{uu^T}{\|u\|^2}\big)S^T
$
is symmetric and positive definite.
\end{proposition} 
\begin{proof} 
The symmetricity is obvious. Let 
$x\in\R^d\setminus\{0\}$, 
denote $\tilde{u}\defeq \frac{u}{\|u\|}$ and define 
$z \defeq S \tilde{u}$. We may write $M = S S^T + a z z^T$, whence
\[
    x^T M x = \| x^T S\|^2 + a (x^T z)^2
    = \| x^T S\|^2\left(
    1 + a \frac{(x^T z)^2}{\| x^T S\|^2}
    \right).
\]
This already establishes the claim in the case $a\ge 0$.
Suppose then $a\in(-1,0)$.
Clearly $(x^T z)^2 = \|x^T S \tilde{u}\|^2 \le \|x^T S\|^2$ and so
$x^T M x \ge \| x^T S\|^2(1-|a|)>0$. 

\end{proof} 

Let us then see what happens in the adaptation in intuitive terms.
Observe first that in (R\ref{step:proposal}) the proposal $Y_n$ is
formed by adding an increment $W_n\defeq S_{n-1} U_n$ to the previous
point $X_{n-1}$. Since $U_n$ is distributed according to the
spherically symmetric $q$, the random variable $W_n$ is distributed
according to the elliptically symmetric density $q_{S_{n-1}}(w) \defeq
\det(S_{n-1})^{-1}q(S_{n-1}^{-1} w)$ with the main axes defined by the
eigenvectors and the corresponding eigenvalues of the matrix $S_{n-1}
S_{n-1}^T$.

To illustrate the behaviour of the RAM update (R\ref{step:adapt}), 
Figure \ref{fig:ellipsoids} shows 
two examples how the contours of the proposal change in the update.
\begin{figure} 
\includegraphics{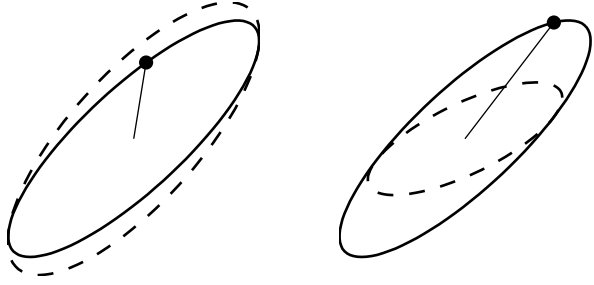}
\caption{Two examples of the RAM update (R\ref{step:adapt}). The solid
  line represents the contour ellipsoid defined by $S_{n-1} S_{n-1}^T$, and the
  vector $S_{n-1} U_n/\|U_n\|$ is drawn as a dot. The contours defined by
  $S_n S_n^T$ are dashed.} 
\label{fig:ellipsoids}
\end{figure}
The example on the left shows how the contour 
ellipsoid expands to the direction 
of $S_n U_n$ when
$\eta_n(\alpha_n-\alpha_*)=0.8>0$. Similarly, the example on the right shows
how the ellipsoid shrinks when $\eta_n(\alpha_n-\alpha_*)=-0.8<0$.
These examples reflect the basic idea behind the approach.
If the acceptance probability is smaller than desired, $\alpha_n<\alpha_*$
(or more than desired, $\alpha_n>\alpha_*$)
the proposal distribution is shrunk (or expanded) 
with respect to the direction of the current proposal increment.

We can also see this behaviour from the update equation
by considering the radius of the contour 
ellipsoid defined by $S_n S_n^T$ with
respect to different directions.
Let $v\in\R^d$ be a unit vector.
As in the proof of Proposition
\ref{prop:posdef}, we may write
\[
    \|S_n^T v\|^2 = \| S_{n-1}^T v\|^2 + \eta_n(\alpha_n-\alpha_*) ( Z_n^T v)^2
\]
where $Z_n = S_n U_n/\|U_n\|$.
If $Z_n$ and $v$ are orthogonal, the latter term
vanishes and $\|S_n^T v\| = \|S_{n-1}^T v\|$. If they are parallel, that is, 
$v = \pm Z_n/\|Z_n\|$, then the factor $(Z_n^T v)^2$ equals
$\|S_{n-1}^T v\|^2$, and so $\|S_n^T v\| = 
\sqrt{1+\eta_n(\alpha_n-\alpha_*)}\|S_{n-1}^T v\|$.
Any other choices of the unit vector $v$ fall in between these two extremes.

\begin{remark} 
    \label{rem:asm} 
    In dimension one, the value of $S_n$ can be computed directly by
\[
    \log S_n = \log S_{n-1} + \frac{1}{2}\log\big(1+\eta_n(\alpha_n -
    \alpha_*)\big).
\]
When $\eta_n$ is small, this is almost equivalent to the update
\[
    \log S_n = \log S_{n-1} + \frac{\eta_n}{2}(\alpha_n - \alpha_*)
\]
implying that the RAM algorithm will exhibit a similar
behaviour with the ASM algorithm as proposed by 
\cite{atchade-fort} and \cite{andrieu-thoms} and analysed by
\cite{vihola-asm}. Therefore, it is justified to consider RAM as
a multidimensional generalisation of the ASM adaptation rule.
\end{remark} 

\begin{remark} 
    In practice, the matrix $S_n$ in (R\ref{step:adapt})
    can be computed as a rank one Cholesky update or downdate of $S_{n-1}$
    when $\alpha_n-\alpha_*>0$ and $\alpha_n-\alpha_*<0$, 
    respectively \citep{linpack-guide}.
Therefore, the algorithm is computationally efficient up to a relatively high 
dimension. In fact, the full $d$-dimensional matrix multiplication 
required when generating the proposal in (R\ref{step:proposal}) has
the same $O(d^2)$ complexity as the Cholesky update or downdate, 
rendering the adaptation to only add a
constant factor to the complexity of the RWM algorithm.
\end{remark} 

\begin{remark} 
While the step size sequence $\eta_n$ can be chosen quite freely, in
practice it is often defined as $\eta_n = n^{-\gamma}$ with
an exponent $\gamma\in(1/2,1]$. The choice $\gamma=1$, which is employed in
the original setting of the AM algorithm \citep{saksman-am} is not advisable
for the RAM algorithm. For simplicity, consider a one-dimensional setting
like in Remark \ref{rem:asm}. Then, if $\eta_n = n^{-1}$ the logarithm of
$S_n$ can increase or decrease only at the speed $\pm\sum_{k=1}^n \eta_k
\approx \log(n)$. Therefore, $S_n$ can grow or shrink only linearly or
at the speed $1/n$, respectively. This renders the adaptation inefficient,
if the initial value $s_1$ differs significantly from the the scale and
shape of $\pi$.
\end{remark} 


\section{Stable points}
\label{sec:convergence} 

The RAM algorithm introduced in the previous section has, under suitable
conditions, a \emph{stable point}, that is, a matrix $S_*\in\R^{d\times d}$,
where the adaptation process $S_n$ should converge as $n$ increases.  Before
considering the convergence, we shall study the stable points of the
algorithm in certain settings.

One can write the update equation
\eqref{eq:ram-update} in the following form
\begin{equation}
    S_n S_n^T = S_{n-1} S_{n-1}^T
    + \eta_n H(S_{n-1},X_{n-1},U_n)
    \label{eq:ashm-rm}
\end{equation}
where
\[
    H(S,x,u)
   =  S \left(
    \min\left\{1,\frac{\pi(x+Su)}{\pi(x)}\right\}
  -\alpha_*\right) \frac{uu^T}{\|u\|^2} S^T.
\]
The recursion \eqref{eq:ashm-rm} implements a so called
Robbins-Monro stochastic approximation algorithm on $(S_n S_n^T)_{n\ge 1}$ 
\citep[e.g.][]{benveniste-metivier-priouret,kushner-yin-sa,borkar-sa}. 
Such an algorithm seeks the root of the
so called mean field $h_{\pi}$ defined as
\begin{equation*}
    h_{\pi}(S) \defeq S
    \int_{\R^d} \int_{\R^d} 
    \left(\min\left\{1,\frac{\pi(x+Su)}{\pi(x)}\right\}    
  -\alpha_*\right) 
  \frac{uu^T}{\|u\|^2}
    q(u) \ud u \pi(x) \ud x S^T.
\end{equation*}
We shall see that under some sufficient conditions, there exists a stable
point, that is, $h_{\pi}(S)=0$.

First, we shall observe a fundamental property of the RAM algorithm;
that it is invariant under affine
transformations.
\begin{theorem} 
    \label{th:process-invariance} 
Let $\pi$ be a probability density and let $(X_n,S_n)_{n\ge 1}$ be the
RAM process (R\ref{step:proposal})--(R\ref{step:adapt}) 
targeting $\pi$ and started from
$(x_1,s_1)$. Suppose $A\in\R^{d\times d}$ is a non-singular matrix,
$b\in\R^d$ and define $\hat{\pi}(x) \defeq |\det(A)|^{-1} \pi(A^{-1} x - b)$.
Let $(\hat{X}_n,\hat{S}_n)_{n\ge 1}$ be the RAM process 
targeting $\hat{\pi}$ and started from $(A x_1+b, A s_1)$. Then,
the processes $(AX_n+b,(AS_n) (AS_n)^T )_{n\ge 1}$ and
$(\hat{X}_n,\hat{S}_n\hat{S}_n^T)_{n\ge 1}$ have identical distributions.
\end{theorem} 
\begin{proof} 
Let $U_n\sim q$ and $W_n\sim U(0,1)$ be the independent sequences 
that drive the RAM process $(X_n,S_n)_{n\ge 1}$ targeting $\pi$; that is
\begin{eqnarray}
    Y_{n} &=& X_{n-1} + S_{n-1} U_{n} \label{eq:y-rec-invar} \\
    X_{n} &=& Y_{n} \charfun{\{W_{n}\le \alpha_n\}}
                + X_n   \charfun{\{W_{n}> \alpha_n\}}
                \label{eq:x-rec-invar}.
\end{eqnarray}
The proof proceeds by constructing an independent sequence 
$\hat{U}_n\sim q$, so that the RAM process 
$(\tilde{X}_n,\tilde{S}_n)_{n\ge 1}$ targeting $\tilde{\pi}$
and driven by $(\tilde{U}_n)_{n\ge 1}$ and $(W_n)_{n\ge 1}$
will satisfy the claim path-wise:
$A X_n = \hat{X}_n$ and $A S_n (A S_n)^T = \hat{S}_n \hat{S}_n^T$
for all $n\ge 1$.

Write the QR decomposition $(A S_n)^T = Q_n R_n$ where $Q_n$ is orthogonal
and where $\hat{S}_n\defeq R_n^T$ is lower-diagonal and chosen so that it has
a positive diagonal.
We observe that $A S_n (A S_n)^T = \hat{S}_n \hat{S}_n^T$
and defining $\hat{U}_{n+1}
\defeq Q_n^T U_{n+1}$ we have also $A S_n U_{n+1} = \hat{S}_n \hat{U}_{n+1}$.
Since the distribution of $U_{n+1}$ is spherically symmetric and $U_{n+1}$ is 
independent of $Q_n$, the sequence $(\tilde{U}_n)_{n\ge 1}$ is i.i.d.~with
distribution $q$.

Now, we may verify inductively
using \eqref{eq:y-rec-invar} and
\eqref{eq:x-rec-invar} that $\hat{X}_n=A X_n$ can be computed
through
\begin{eqnarray*}
    \hat{Y}_{n} &=& \hat{X}_{n-1} + \hat{S}_{n-1} \hat{U}_{n} \\
    \hat{X}_{n} &=& \hat{Y}_{n} \charfun{\{W_{n}\le \hat{\alpha}_n\}}
                + X_{n-1}   \charfun{\{W_{n}> \hat{\alpha}_n\}}
\end{eqnarray*}
where 
\[
    \hat{\alpha}_n
    = \min\bigg\{1,\frac{\hat{\pi}(\hat{Y}_{n})
       }{\hat{\pi}(\hat{X}_{n-1})}\bigg\} 
    \!= \min\left\{1,\frac{\pi(Y_n)}{\pi(X_{n-1})}\right\}
    \!= \alpha_n. \qedhere
\]
\end{proof} 
After Theorem \ref{th:process-invariance}, it is no surprise that the mean
field of the algorithm satisfies similar invariance properties.
\begin{theorem} 
    \label{th:invariance} 
Suppose $\pi$ is a probability density.
\begin{enumerate}[(i)]
   \item \label{item:affine}
     Let $\hat{\pi}$ be an affine transformation of $\pi$, that
is, $\hat{\pi}(x) = |\det(A)|^{-1} \pi(A^{-1}x-b)$ for some
non-singular matrix $A\in\R^{d\times d}$ and $b\in\R^d$. 
Then, $A h_{\pi}(S) A^T = h_{\hat{\pi}}(AS)$ for all $S\in \R^{d\times d}$.
   \item \label{item:rotation}
     For any orthogonal matrix $Q\in\R^{d\times d}$ and for all
     $S\in\R^{d\times d}$,
     $h_{\pi}(S) = h_{\pi}(SQ)$.
   \item \label{item:uniq-invar}
     Suppose that $S$ is a unique lower-diagonal matrix with positive diagonal
   satisfying $h_{\pi}(S)=0$. Then, restricted to such matrices, the solution of
   $h_{\hat{\pi}}(\hat{S})=0$ is also unique, and of the form
   $\hat{S}=AS Q$ for some orthogonal $Q\in\R^{d\times d}$.
\end{enumerate}
\end{theorem} 
\begin{proof} 
The claim (\ref{item:affine}) follows by a change of variable $x = A^{-1}z -b $, 
\[
  \begin{split}
    h_{\pi}(S) 
    &= S
    \int_{\R^d} \int_{\R^d} 
    \left(\min\left\{1,\frac{\pi(x+Su)}{\pi(x)}\right\}
  -\alpha_*\right) 
  \pi(x) \ud x \frac{uu^T}{\|u\|^2}
    q(u) \ud u S^T \\
    &= S
    \int_{\R^d} \int_{\R^d} 
    \left(\min\left\{1,\frac{\hat{\pi}(z + ASu)}{\hat{\pi}(z)}\right\}
  -\alpha_*\right) \hat{\pi}(z) \ud z \frac{uu^T}{\|u\|^2}
    q(u) \ud u S^T \\
    &= A^{-1}h_{\hat{\pi}}(AS)A^{-T}.
\end{split}
\]
The claim (\ref{item:rotation}) follows from similarly, by a change of 
variable $u=Qv$
and due to the spherical symmetry of $q$.
The uniqueness up to rotations, that is, only the matrices of the form
$\hat{S}=ASQ$ satisfy $h_{\hat{\pi}}(\hat{S})=0$ follows directly as above.
The claim (\ref{item:uniq-invar}) is completed by writing the
QR-decomposition $(AS)^T=QR$. and by observing that the upper-triangular 
$R$ can be chosen to have positive diagonal elements.
\end{proof} 

Theorem \ref{th:invariance} verifies that the stable points of the algorithm
are affinely invariant like the covariance (or more generally robust
pseudo-covariance) matrices \citep{huber}. Theorem
\ref{th:fixed-point-elliptic} below verifies that in the case of a suitable
elliptically symmetric target distribution $\pi$, the stable points of the
RAM algorithm in fact coincide with the (pseudo-)covariance of $\pi$. This
is an interesting connection, but in general the fixed points of the RAM
algorithm are not expected to coincide with the pseudo-covariance.

\begin{theorem} 
    \label{th:fixed-point-elliptic} 
Assume $\alpha_*\in(0,1)$ and 
$\pi$ is elliptically symmetric, that is, $\pi(x) \equiv
\det(\Sigma)^{-1} p(\|\Sigma^{-1} x\|)$ for some $p:[0,\infty)\to[0,\infty)$
and for some symmetric and positive definite $\Sigma\in\R^{d\times d}$.
Then,
\begin{enumerate}[(i)]
    \item \label{item:existence} 
      there exists a lower-diagonal matrix with positive diagonal
$S_*\in\R^{d\times d}$ such that $h_{\pi}(S_*)=0$ and such that $S_*S_*^T$
is proportional to
$\Sigma^2$.
\item \label{item:uniqueness} 
  assuming the function $p$ is non-increasing, the solution $S_*$
  is additionally unique.
\end{enumerate}
\end{theorem} 
\begin{proof} 
In light of Theorem \ref{th:invariance}, it is sufficient to consider
any spherically symmetric $\pi$, that is, the case $\Sigma$ is an identity
matrix.

Let $S$ be a lower-diagonal matrix with positive diagonal.
Observe that since $S$ is non-singular, 
$h_{\pi}(S)=0$ is equivalent to $S^{-1} h_{\pi}(S)
S^{-T}=0$, that is
\begin{equation}
    \int_{\R^d} \int_{\R^d} 
    \left(\min\left\{1,\frac{\pi(x+Su)}{\pi(x)}\right\}
  -\alpha_*\right) 
  \frac{uu^T}{\|u\|^2}
    q(u) \ud u \pi(x) \ud x = 0.
    \label{eq:stable-point}
\end{equation}
Define the function
\begin{equation*}
    \bar{h}(S) \defeq 
    \int_{\R^d} \int_{\R^d} 
    \left(\min\left\{1,\frac{\pi(x+Su)}{\pi(x)}\right\} \right) 
     \frac{uu^T}{\|u\|^2}
    q(u) \ud u \pi(x) \ud x.
\end{equation*}
It is easy to see by symmetry and taking traces 
that \eqref{eq:stable-point} is equivalent to
$\bar{h}(S)=\frac{\alpha_*}{d} I$, where $I\in\R^{d\times d}$ stands for the
identity matrix.

We can write $\bar{h}(S)$ in a more convenient form by
using the polar coordinate representation $u = rv$, where 
$v\in \mathcal{S}^d \defeq \{v\in\R^d: \|v\|=1 \}$ is a unit vector in 
the unit sphere, and $r=\|u\|$ is the length of $u$. Then, by 
Fubini's theorem 
\begin{equation*}
    \bar{h}(S) =  \int_{\mathcal{S}^d} 
    \left[
    \int_0^\infty 
    \int_{\R^d}
    \min\left\{\pi(x),\pi(x+rSv)\right\} \ud x
     \hat{q}(r) \ud r \right]
    vv^T \mu(\ud v)
\end{equation*}
where $\mu$ stands for the uniform distribution 
on the unit sphere $\mathcal{S}^d$
and the proposal is written as $q(u)\propto \hat{q}(\|u\|)$.

By applying the representation of $\pi$ by the radial function $p$
one can write the term above in brackets as
\[
    g(\|Sv\|) \defeq
    \int_0^\infty 
    \int_{\R^d}
    \min\left\{p(\|x\|),p(\|x+r S v\|)\right\} \ud x
   \hat{q}(r) \ud r,
\]
since due to symmetry, 
the value of the integral depends only on the norm $\|Sv\|$.

For any $\theta\in\R_+$, one can now write
\[
    \bar{h}(\theta I)
    = \int_{\mathcal{S}^d} g(\theta) v v^T \mu(\ud v)
    = \frac{g(\theta)}{d}I,
\]
since $\trace\big(\bar{h}(\theta I)\big) = g(\theta)$ and by symmetry.
Proposition \ref{prop:convolutions} in Appendix \ref{app:convolutions}
shows that $g:(0,\infty)\to(0,\infty)$ is continuous, 
that $\lim_{\theta\to\infty}g(\theta) = 0$ 
and that $\lim_{\theta\to 0+} g(\theta)= \int_0^\infty \hat{q}(r)\ud r=1$. 
Therefore, there exists a $\theta_*>0$ such that $g(\theta_*)=\alpha_*$ so 
that  $\bar{h}(\theta_* I)=\frac{\alpha_*}{d} I$, establishing (\ref{item:existence}).

For (\ref{item:uniqueness}), let us first show that $g$ is in this case 
strictly decreasing, at least before hitting zero.
Observe that since $p$ is non-increasing, one can write
\[\begin{split}
    g(\theta) &= \int_0^\infty
     \bigg( \int_{\|x\|>\|x+r\theta v\|}
     p(\|x\|) \ud x 
+ \int_{\|x\|\le \|x+r\theta v\|} p(\|x+r\theta v\|) \ud x 
\bigg) \tilde{q}(r)\ud r \\
&= \int_0^\infty
     \bigg( 1 - 
     \int_{A_{r\theta v}}
     \pi(x) \ud x
\bigg) \tilde{q}(r)\ud r. 
\end{split}\]
It is easy to see that the width of the strip $A_{r\theta v}\defeq 
\{\|x\|\le \|x+r\theta v\|\} \cap 
\{\|x\|<\|x-r\theta v\|\}$ is increasing with 
respect to $\theta$. Therefore, for any fixed $r$ and $v$, 
the term $b_{rv}(\theta)\defeq 1-\int_{A_{r\theta v}}
  \pi(x) \ud x$ is strictly decreasing with respect to $\theta$ 
  as long as the support of $\pi$ is not completely covered by
  $A_{r\theta v}$, in which case $b_{rv}(\theta)=0$.
This implies that $g(\theta)$ is strictly decreasing with
respect to $\theta$, until possibly $g(\theta)=0$.
Therefore, there is a unique $\theta_*>0$ for which
$g(\theta_*)=\alpha_*$.

Let us assume that $S\in\R^{d\times d}$ is a matrix satisfying 
$\bar{h}(S)=\frac{\alpha_*}{d} I$. By symmetry, we can 
assume $S$ to be diagonal, with
positive diagonal elements $s_1,\ldots,s_d>0$.
Let $e_1,\ldots,e_d$ stand for the standard basis vectors of $\R^d$.
The diagonal element $[\bar{h}(S)]_{ii}=\frac{\alpha_*}{d}$ is equivalent to
\[
    \int_{\mathcal{S}^d} \left[g(\|S v\|)-\alpha_*\right] (v^T e_i)^2
    \mu(\ud v) = 0,
\]
since $\int_{\mathcal{S}^d} (v^T e_i)^2 \mu(\ud v)=d^{-1}$.
Denoting $\bar{g}(\|S v\|) \defeq g(\|S v\|)-\alpha_*$, this implies
\begin{equation}
    \int_{\mathcal{S}^d} 
    \bar{g}\Big(\big(\textstyle\sum_{i=1}^d s_i^2 v_i^2\big)^{1/2}\Big)
    \Big(\textstyle\sum_{i=1}^d \lambda_i v_i^2\Big)
    \mu(\ud v) = 0
    \label{eq:lin-comb}
\end{equation}
for any choice of the constants $\lambda_i\in\R$. Particularly, choosing
$\lambda_i=1$ for $i=1,\ldots,d$ implies that for any constant $c\in\R$ we 
have
\begin{equation}
    \int_{\mathcal{S}^d} 
    \bar{g}\Big(\big(\textstyle\sum_{i=1}^d s_i^2 v_i^2\big)^{1/2}\Big)
    c \mu(\ud v) = 0.
    \label{eq:const}
\end{equation}
Now, summing \eqref{eq:lin-comb} and \eqref{eq:const} 
with a specific choice of constants $c=\theta_*^2$ and $\lambda_i=-s_i^2$,
we obtain
\[
    \int_{\mathcal{S}^d} 
    \bar{g}\Big(\big(\textstyle\sum_{i=1}^d s_i^2 v_i^2\big)^{1/2}\Big)
    \Big(\theta_*^2 - \textstyle\sum_{i=1}^d s_i^2 v_i^2\Big)
    \mu(\ud v) = 0.
\]
But now, $\bar{g}\big( (\textstyle\sum_{i=1}^d s_i^2 v_i^2)^{1/2}\big)
\ge 0$ exactly when $\textstyle\sum_{i=1}^d s_i^2 v_i^2\le \theta_*^2$,
so the integrand is always non-negative. 
Moreover, if any $s_i\neq \theta_*$, then by continuity
there is a neighbourhood $U_i\subset\mathcal{S}^d$ of $e_i$ 
such that the integrand is strictly positive, implying that the integral is
strictly positive. This concludes 
the proof of the uniqueness (\ref{item:uniqueness}).
\end{proof} 

The following theorem shows that when 
$\pi$ is the joint density of $d$ independent and identically
distributed random variables, the RAM algorithm has, as expected, a stable point
proportional to the identity matrix.
\begin{theorem} 
    \label{th:fixed-point-iid} 
Assume $\alpha_*\in(0,1)$ and 
$\pi(x)=\prod_{i=1}^d p(x_i)$
for some one-dimensional density $p$.
Then, there exists a $\theta>0$ such that $\hat{h}(\theta I)=0$.
\end{theorem} 
\begin{proof} 
    Let $e_1,\ldots,e_d$ stand for the coordinate vectors of $\R^d$.
    Consider the functions
    \begin{equation*}
    a_i(\theta)\defeq 
    \int_{\mathcal{S}^d}
    \int_0^\infty 
    \left(\int_{\R^d}
\min\left\{\pi(x),\pi(x+r\theta u)\right\}    
    \ud x \right)  \hat{q}(r) \ud r
    (u^T e_i)^2
    \mathcal{H}^{d-1}(\ud u).
    \end{equation*}
    Let $P$ be a permutation matrix.
    It is easy to see that $\pi(x+r\theta u) = \pi\big(P(x+r\theta u)\big)$ 
    by the i.i.d.~product form of $\pi$. Therefore, by the change of
    variable $Px =z$ and $Pu=v$, one obtains that
    \begin{multline*}
    a_i(\theta) =
    \int_{\mathcal{S}^d}
    \int_0^\infty 
    \left(\int_{\R^d}
\min\left\{\pi(z),\pi(z+r\theta v)\right\}    
    \ud x \right) \\
    \times \hat{q}(r) \ud r
    (v^T P^T e_i)^2
    \mathcal{H}^{d-1}(\ud v)
    = a_j(\theta)
    \end{multline*}
    by a suitable choice of $P$.
    Moreover, $\lim_{\theta\to\infty} a_i(\theta)=0$ and
    $\lim_{\theta\to 0+} a_i(\theta)= c \defeq \int_{\mathcal{S}^d} (u^T e_i)^2
    \mathcal{H}^{d-1}(\ud u)$ and $a_i$ are continuous. Therefore,
    there exists a $\theta_*>0$ such that $a_i(\theta_*)=
    a_*c $, and so 
    $e_i^T h(\theta_* I) e_i=0$. 
    
    It remains to show that $e_i h(\theta_* I) e_j = 0$ for all $i\neq j$.
    But for this, it is enough to show that
    the integrals of the form
    \begin{equation*}
        \int_{E_{i,j}^*} 
    \int_0^\infty 
    \left(\int_{\R^d}
\min\left\{\pi(z),\pi(z+r\theta v)\right\}
    \ud x \right)
    \hat{q}(r) \ud r
    |(v^T e_i)(v^T e_j)|
    \mathcal{H}^{d-1}(\ud v)
    \end{equation*}
    have the same value for both $E_{i,j}^{+} 
    \defeq \{v\in\mathcal{S}^d:(v^T e_i)(v^T e_j)> 0\}$ and
    $E_{i,j}^{-}\defeq \{v\in\mathcal{S}^d:(v^T e_i)(v^T e_j)< 0\}$.
    But this is obtained due to the symmetry of the sets $E_{i,j}^{+}$ and
    $E_{i,j}^{-}$ and the product form of $\pi$, since
    \begin{equation*}
        \int_{\R^d}
\min\left\{\pi(z),\pi(z+r\theta v)\right\}
    \ud x
    = \int_{\R^d}
\min\left\{\textstyle\pi\big(z-\frac{1}{2}r\theta v\big),\pi\big(z+\frac{1}{2}r\theta
  v\big)\right\}
    \ud x
    \end{equation*}
    so one can change the sign of any coordinate of $v$ without
    affecting this integral.
    This concludes the claim.
\end{proof} 

\begin{remark} 
    Checking the existence and uniqueness in a more general setting
    it is out of the scope of this paper.
    It is believed that there always exists at least one solution
    $S_*\in\R^{d\times d}$ such that $h(S_*)=0$.
    Notice, however, that the fixed point may not 
    be always unique; see an example of such a situation for 
    one-dimensional adaptation (the ASM algorithm) in
    \cite[Section 4.4]{hastie-phd}.
\end{remark} 

\begin{remark} 
\label{rem:stable-drift} 
It is not very difficult to show that for any given target $\pi$ and
proposal $q$, there exist some constants $0<\theta_1<\theta_2<\infty$ such
that the matrices $h_\pi(\theta_1 I)$ and $h_\pi(\theta_2 I)$ are positive
definite and negative definite, respectively. This indicates that, on
average, $S_n$ should shrink whenever it is `too big' and expand whenever
it is `too small,' so the algorithm should admit a stable behaviour. 
The empirical results in Section \ref{sec:tests} support the hypothesis of
general stability.
\end{remark} 

To be more precise, we can identify a Lyapunov function $w_\pi$
for $h_\pi$ in the case $\pi$ is elliptically symmetric with a
non-increasing tail. This will allow us to establish the convergence
of the sequence $(S_n S_n^T)_{n\ge 1}$ in Theorem \ref{th:convergence}.
\begin{theorem} 
    \label{th:lyapunov} 
Assume the conditions of Theorem \ref{th:fixed-point-elliptic} (ii)
and denote $R_* \defeq S_* S_*^T$.
Define a function $w_\pi:\R^{d\times d}\to[0,\infty)$ by
\[
    w_\pi(R) \defeq 
    \trace(R_*^{-1} R) - \log\left(\frac{\det R}{\det R_*}\right) - d.
\]
Then, for any non-singular $S\in\R^{d\times d}$ it holds that $\big\langle
\nabla w_\pi(S S^T),h_\pi(S)\big\rangle \le 0$ with equality only if $S S^T=R_*$.
\end{theorem} 
\begin{proof} 
Denote $\hat{\pi}(x) \defeq \det(R_*)^{1/2} \pi(R_*^{1/2} x)$, then by Theorem
\ref{th:invariance} (\ref{item:affine}) $h_{\pi}(S) = R_*^{1/2}
h_{\hat{\pi}}(R_*^{-1/2} S)R_*^{1/2}$.
Moreover, Theorem \ref{th:fixed-point-elliptic} (ii) together with
Theorem \ref{th:invariance}
(\ref{item:uniq-invar}) imply that $\hat{\pi}$ is spherically symmetric
and $S=I$ is the unique solution of 
$h_{\hat{\pi}}(S)=0$ (up to orthogonal transformations). 

We can write
\begin{align*}
    \nabla w_\pi\big(R_*^{1/2} S (R_*^{1/2} S)^T\big) 
&= R_*^{-1/2}(I - (SS^T)^{-1})R_*^{-1/2}  
= R_*^{-1/2} \nabla w_{\hat{\pi}}(S) R_*^{-1/2}, 
\end{align*}
so we obtain
\begin{align*}
    \big\langle
    \nabla w_\pi\big(R_*^{1/2} S (R_*^{1/2} S)^T\big),
    h_\pi(R_*^{1/2} S\big\rangle 
    &= \trace\big[\nabla w_\pi\big(R_*^{1/2} S (R_*^{1/2}
    S)^T\big)^T
    h_\pi(R_*^{1/2})\big] \\
    &=\big\langle
    \nabla w_{\hat{\pi}}(S),
    h_{\hat{\pi}}(S)\big\rangle.
\end{align*}
Therefore, it is sufficient to check that the claim holds 
for spherically symmetric $\hat{\pi}$ with $R_* = I$. 

Let $S$ be non-singular and write the singular value decomposition 
$S=U\bar{S}V^T$ where $U$ and $V$ are orthogonal and
$\bar{S}=\diag(\bar{s}_1,\ldots,\bar{s}_d)$ with positive
diagonal entries. By 
Theorem \ref{th:invariance} (\ref{item:rotation})
we have $h_{\hat{\pi}}(S) = h_{\hat{\pi}}(SV) = h_{\hat{\pi}}(U\bar{S})$.
We may write, using the notation in Theorem
\ref{th:fixed-point-elliptic},
\begin{align*}
\trace\big(h_{\hat{\pi}}(S)\big)
&=\trace\big(U^Th_{\hat{\pi}}(U\bar{S}) U\big) 
= \int_{\mathcal{S}^d} 
\bar{g}\big(\| \bar{S} w \|\big)\left[
{\textstyle\sum_{i=1}^d }\bar{s}_i^2 w_i^2 \right]\mu(\ud w).
\end{align*}
We have $SS^T = U \bar{S}^2 U^T$, so we obtain similarly
\begin{align*}
\trace\big((SS^T)^{-1} h_{\hat{\pi}}(S)\big) &= 
\trace\big(\bar{S}^{-1} U^T h_{\hat{\pi}}(S V) U \bar{S}^{-1}\big) 
=\int_{\mathcal{S}^d} \bar{g}\big(\| \bar{S} w \|\big)
\mu(\ud w).
\end{align*}
Putting everything together,
\begin{equation*}
    \big\langle
\nabla w_{\hat{\pi}}(SS^T),h_{\hat{\pi}}(S)\big\rangle 
 =\int_{\mathcal{S}^d} 
\bar{g}\Big(\big( {\textstyle \sum_{i=1}^d }
\bar{s}_i^2 w_i^2\big)^{1/2}\Big)\Big(
{\textstyle\sum_{i=1}^d }\bar{s}_i^2 w_i^2 -1 \Big)\mu(\ud w).
\end{equation*}
As in the proof of Theorem \ref{th:fixed-point-elliptic}, 
$\bar{g}\big( ( {\textstyle \sum_{i=1}^d }\bar{s}_i^2
w_i^2)^{1/2}\big) >0$ exactly when $\sum_{i=1}^d \bar{s}_i^2
w_i^2< 1$ and vice versa. The integral can equal zero only if 
all $\bar{s}_i=1$.
\end{proof} 


\section{Validity}
\label{sec:lln} 

This section describes some sufficient conditions under which the RAM
algorithm is valid; that is, when the empirical averages converge
to the integral
\begin{equation}
    I_n = \frac{1}{n}\sum_{k=1}^n f(X_k) \xrightarrow{n\to\infty}
    \int_{\R^d} f(x) \pi(x) \ud x \eqdef I
    \label{eq:slln}
\end{equation}
almost surely.

Let us start by introducing assumptions on the forms of the proposal 
density $q$ and the target density $\pi$.
\begin{assumption} 
    \label{a:proposal} 
The proposal density $q$ is either a Gaussian or a Student
distribution, that is,
    \[
        q(z) \propto e^{-\frac{1}{2}\|z\|^2}
        \qquad\text{or}\qquad
        q(z) \propto (1 + \|z\|^2)^{-\frac{d+p}{2}}
    \]
    for some constant $p>0$.
\end{assumption} 
\begin{assumption} 
    \label{a:target} 
The target density $\pi$ satisfies either of the following assumptions.
\begin{enumerate}[(i)]
    \item \label{item:compact} 
      The density $\pi$ is bounded and supported on a bounded set: there
exists a constant $m<\infty$ such that $\pi(x)=0$ for all $\|x\|\ge m$.
   \item \label{item:super-exp}
    The density $\pi$ is positive everywhere in $\R^d$ 
    and continuously
    differentiable. The tails of $\pi$ are super-exponentially decaying and
    have regular contours, that is, respectively
    \begin{eqnarray*}
        \lim_{\|x\|\to\infty} \frac{x}{\|x\|} \cdot \nabla \log \pi(x) 
        &=& -\infty
        \qquad\text{and}\\
        \limsup_{\|x\|\to\infty} 
        \frac{x}{\|x\|} \cdot \frac{\nabla \pi(x)}{\|\nabla \pi(x)\|} &<& 0.
    \end{eqnarray*}
\end{enumerate}
\end{assumption} 
\begin{remark} 
    Assumption \ref{a:target} 
    ensures the geometric ergodicity of
    the RWM algorithm under fairly general settings; 
    \cite{jarner-hansen} discuss the limitations of
    (\ref{item:super-exp}) and give several
    examples.
\end{remark} 
Before stating the theorem, consider the following 
conditions on the adaptation step size sequence
$(\eta_n)_{n\ge 1}$ and on the stability of the process $(S_n)_{n\ge 1}$.
\begin{assumption} 
    \label{a:step-size} 
The adaptation step sizes $\eta_n\in[0,1]$ are non-increasing and satisfy
$\sum_{n=1}^\infty k^{-1} \eta_n < \infty$.
\end{assumption} 
\begin{assumption} 
    \label{a:stability} 
There exist random variables $0\le A\le B\le \infty$ such that all 
the eigenvalues
$\lambda_n^{(i)}$ of the random matrices $S_n S_n^T$ are almost surely
bounded by $A\le \lambda_n^{(i)} \le B$, for all $n=1,2,\ldots$ and all
$i=1,\ldots,d$.
\end{assumption} 
\begin{theorem} 
\label{th:slln} 
Suppose Assumptions \ref{a:proposal}--\ref{a:stability} hold
and denote $\Omega_0\defeq \{A>0,\, B<\infty\}$.
Suppose also that the function $f:\R^d\to\R$ satisfies
for some $p\in[0,1)$
\[
    \sup_{x\in\R^d:\pi(x)>0} |f(x)| \pi^{-p}(x)<\infty.
\]
Then, for almost every $\omega\in\Omega_0$, the strong law of
large numbers \eqref{eq:slln} holds.
\end{theorem} 
The proof follows
by existing results in the literature; the details 
are given in Appendix \ref{sec:slln-proof}.

The convergence of the adaptation can also be established in case
$\pi$ is elliptically symmetric.
\begin{theorem} 
\label{th:convergence} 
  If the conditions of Theorem \ref{th:fixed-point-elliptic}
  (\ref{item:uniqueness}) and Theorem \ref{th:slln} hold and
  additionally $\sum_n \gamma_n = \infty$, then $S_{n} S_{n}^T \to S_*
  S_*^T$ for almost every $\omega\in\Omega_0$.
\end{theorem} 
The proof follows by Theorem \ref{th:lyapunov} and 
results in the literature; see 
Appendix \ref{sec:slln-proof}.

\begin{remark} 
Assumptions \ref{a:proposal}--\ref{a:step-size} are common when verifying
the ergodicity of an adaptive MCMC algorithm.
Assumption \ref{a:stability} on stability is natural but it
can be difficult to check with $\P(A>0,\,B<\infty)=1$ in practice. 
The empirical evidence supports this hypothesis under a very 
general setting;
see also Remark \ref{rem:stable-drift} in Section
\ref{sec:convergence}. 
Similar stability results have been established
only for few adaptive MCMC algorithms, including the AM and the
ASM algorithms~\citep{saksman-vihola,vihola-collapse,vihola-asm}. 
The precise stability analysis is beyond the scope of this
paper. Instead, the stability can be enforced as described below.

Let $0<a\le b<\infty$ be some constants so that the eigenvalues of $s_1 s_1^T$
are within $[a,b]$. Then, replace the step (R\ref{step:adapt}) in the RAM
algorithm with the following:
\begin{enumerate}[({R}1')]
    \setcounter{enumi}{2}
  \item compute the lower-diagonal matrix $\hat{S}_n$ with positive diagonal 
    so that $\hat{S}_n \hat{S}_n^T$ equals the right hand side of
    \eqref{eq:ram-update}.
    If the eigenvalues of $\hat{S}_n \hat{S}_n^T$ are within $[a,b]$, then 
    set $S_n = \hat{S}_n$, otherwise set $S_n = S_{n-1}$.
\end{enumerate}
While this modification ensures stability, it may change the stable
points of the algorithm and the conclusion of Theorem \ref{th:convergence} 
may not hold. This could possibly be avoided, for example, by
considering an adaptive reprojections approach
\citep{sa-verifiable,andrieu-moulines}, but we do not pursue this here.
\end{remark} 


\section{Experiments}
\label{sec:tests} 


The RAM algorithm was tested with three types of target distributions:
heavy-tailed Student, Gaussian and a mixture of Gaussians.
The performance of RAM was compared against
the seminal adaptive Metropolis (AM) algorithm \citep{saksman-am} and an
adaptive scaling within adaptive Metropolis (ASWAM) algorithms
\citep{andrieu-thoms,atchade-fort}. Especially the comparison against ASWAM
is of interest, since it attains a given acceptance rate like the RAM
algorithm. 

There are several parameters that are fixed throughout the experiments. The
adaptation step size sequence was set to $\eta_n = n^{-2/3}$ for the AM and the
ASWAM algorithms. For the RAM approach, the weight sequence was modified 
slightly so that $\eta_n = \min\{1,d\cdot n^{-2/3}\}$. The extra factor 
was added to compensate the
expected growth or shrinkage of the eigenvalues being of the order $d^{-1}$; 
see the proof of Theorem
\ref{th:fixed-point-elliptic}. The target mean acceptance rate was 
$\alpha_* = 0.234$. In all the experiments, the Student proposal
distribution of the form $q(z) = (1+\|z\|^2)^{-\frac{d+1}{2}}$ was used.
Such a heavy-tailed proposal was employed in order to have good convergence
properties in case of heavy-tailed target densities
\citep{jarner-roberts-heavytailed}.

All the tests were performed using the publicly available
Grapham software \citep{vihola-grapham}; the latest version of the software
includes an implementation of the RAM algorithm.


\subsection{Multivariate Student distribution}
\label{sec:student} 

The first example is a bivariate
Student distribution with $n=1$ degrees of freedom and
the following location and pseudo-covariance matrix
\[
    \mu=\begin{bmatrix}1\\ 2\end{bmatrix}\qquad\text{and}\qquad
    \Sigma = \begin{bmatrix}
      0.2 & 0.1\\
      0.1 & 0.8
      \end{bmatrix},
\]
respectively. That is, the target density $\pi(x)\propto(1+x^T \Sigma^{-1}
x)^{-3/2}$. Clearly, $\pi$ has no second moments and thereby the empirical covariance
estimate used by AM and ASWAM is deemed to be unstable in this example.

Figure~\ref{fig:adapt-student} shows the results for one hundred runs
of the algorithms. The grey area indicates the interval between the 10\% and
the 90\% percentiles, and the black line shows the median.
The top row shows the logarithm of the first diagonal
element of the matrix $S_n$.
\begin{figure*} 
    \psfragscanon
    \psfrag{0}[c][c]{\small 0}
    \psfrag{1}[c][c]{\small 1}
    \psfrag{2}[c][c]{\small 2}
    \psfrag{3}[c][c]{\small 3}
    \psfrag{5}[c][c]{\small 5}
    \psfrag{10}[c][c]{\small 10}
    \psfrag{15}[c][c]{\small 15}
    \psfrag{2e5}[c][c]{\small$200k$}
    \psfrag{4e5}[c][c]{\small$400k$}
    \psfrag{AM}[c][c]{\small AM}
    \psfrag{ASWAM}[c][c]{\small ASWAM}
    \psfrag{RAM}[c][c]{\small RAM}
    \begin{tabular}{rrr}
      \includegraphics{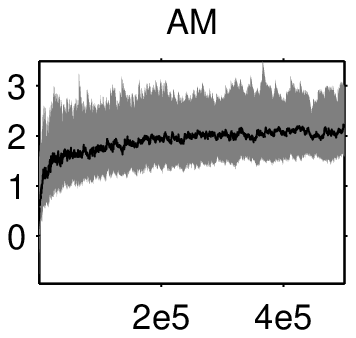}
      & \includegraphics{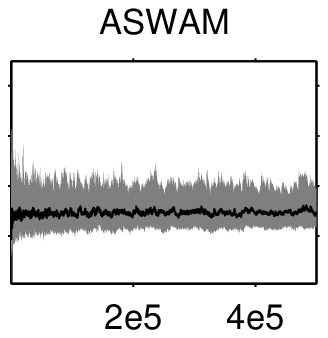}
      & \includegraphics{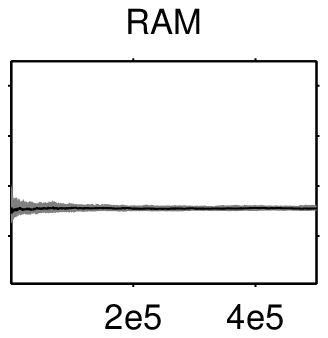}  \\[-20pt]
      \includegraphics{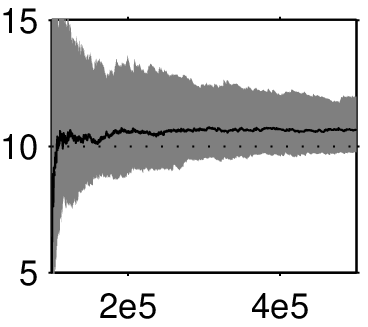}
      & \includegraphics{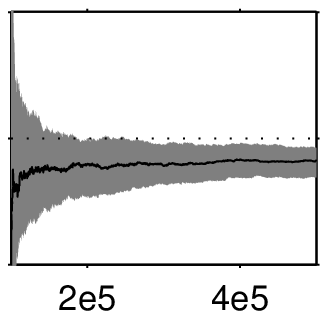}
      & \includegraphics{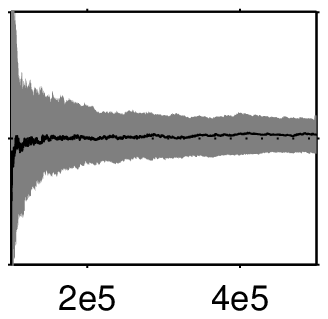} 
    \end{tabular}
    \psfragscanoff
    \vspace{-20pt}
    \caption{Bivariate Student example:
      logarithm of the first diagonal 
      component of the matrix $S_n$ (top) and
      the proportion of $X_n$ in the set $A$ after 100,000 burn-in
      iterations (bottom).} 
    \label{fig:adapt-student}
\end{figure*}
The AM covariance grows without an upper bound as expected.
When the scale adaptation is added, the ASWAM
approach manages to keep the factor $S_n = \theta_n L_n$ within
certain bounds, but there is a considerable variation that does not seem
to vanish.
This is due to the fact that $L_n$, the Cholesky factor of
$\Cov(X_1,\ldots,X_n)$, grows without an upper bound but at the same time the scaling factor 
$\theta_n$ decays to keep the acceptance rate around the desired $23.4\%$.
The RAM algorithm seems to converge nicely to a limiting value.

Such undesided behaviour of the AM and the ASWAM algorithms may also have an
effect on the validity of their simulation. Indeed, let us consider 
the 90\% highest probability density (HPD) set of the target, that is,
the set $A\defeq
\{x\in\R^2: (x-\mu)^T \Sigma^{-1} (x-\mu)^T > 99\}$. Figure
\ref{fig:adapt-student} (bottom) shows the percentage of $X_n$ outside the 
90\% HPD computed after a 100,000 sample burn-in period. The AM algorithm tends to
overestimate the ratio slightly, with more variation than the ASWAM and
the RAM approaches. The estimate produced by the ASWAM algorithm has
approximately the same variation as RAM, but there is a tendency to
underestimate the ratio. The RAM estimates are centred around the true
value.

To check how the RAM algorithm copes with higher dimensions, let us follow
\cite{roberts-rosenthal-examples} and consider
a matrix $\Sigma = M M^T$, where $M\in\R^{d\times
d}$ is randomly generated with i.i.d.~standard Gaussian elements. 
Such a matrix $\Sigma$ is used as the pseudo-covariance of a Student
distribution, so that $\pi(x) \propto (1+x^T
\Sigma^{-1} x)^{-\frac{d+1}{2}}$. 
\cite{roberts-rosenthal-scaling} showed that
in the case of Gaussian target and proposal distributions, one can measure
the `suboptimality' by the factor 
$b \defeq d \big(\sum_{i=1}^d
\lambda_i^{-2}\big)\big(\sum_{i=1}^d \lambda_i^{-1}\big)^{-2}$
where $\lambda_i$ are the eigenvalues of the matrix $(S_n
S_n^T)^{1/2}\Sigma^{-1/2}$. The factor equals one if the matrices are
proportional to each other, and is larger otherwise.
While the factor may not have the same interpretation in the present setting
involving Student distributions, it serves as a good measure of mismatch
between $S_n S_n^T$ and $\Sigma$.
Figure \ref{fig:adapt-student-b} shows the
factor $b$ in increasing dimensions each based on 100 runs of the RAM
algorithm.
\begin{figure*} 
    \psfragscanon
    \psfrag{d=10}{\small$d=10$}
    \psfrag{d=20}{\small$d=20$}
    \psfrag{d=30}{\small$d=30$}
    \psfrag{1.2}{\small 1.2}
    \psfrag{1.4}{\small 1.4}
    \psfrag{1.6}{\small 1.6}
    \psfrag{1.8}{\small 1.8}
    \psfrag{2e5}{\small$200k$}
    \psfrag{5e5}{\small$500k$}
    \psfrag{8e5}{\small$800k$}
    \begin{tabular}{rrr}
      \includegraphics{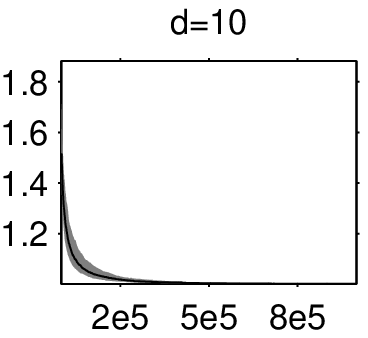}
      & \includegraphics{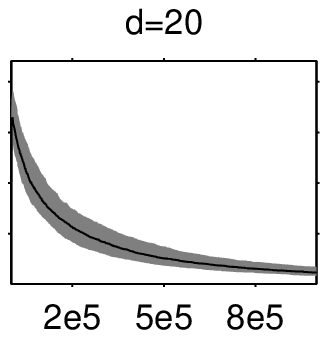}
      & \includegraphics{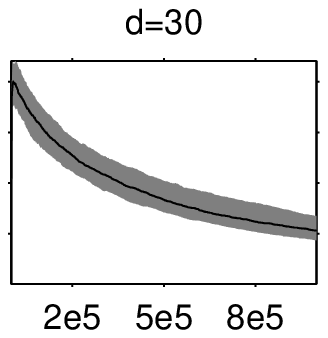}
    \end{tabular}
    \psfragscanoff
    \vspace{-20pt}
    \caption{Suboptimality factor $b$ over one million
      iterations of the RAM algorithm 
      with a different dimensional Student target.} 
    \label{fig:adapt-student-b}
\end{figure*}
The convergence of $S_n S_n^T \to \Sigma$ is slower in higher dimensions,
but the algorithm seems to find a fairly good approximation already with a
moderate number of samples.


\subsection{Gaussian distribution}
\label{sec:gaussian} 

The multivariate Gaussian target $\pi(x) = \mathcal{N}(0,\Sigma)$ serves as
a baseline comparison for the algorithms, as they should converge to the
same matrix factor\footnote{For the AM algorithm, the constant $\theta_*$ is
slightly different, but approximately equal in higher dimensions.} $S_n
S_n^T \to \theta_*\Sigma$. 

The algorithms were tested in different dimensions, for one thousand 
covariance matrices randomly generated as described in Section
\ref{sec:student}. The algorithms were always started in `steady state' so
that $X_1 \sim N(0,\Sigma)$. The algorithms were run half a million
iterations: 100,000 burn-in and 400,000 to estimate the proportions of the
samples $X_n$ in the 10\%, 25\%, 50\%, 75\% and 90\% 
HPD of the distribution. Table \ref{tab:normal-quantiles} shows the
overall root mean square error. 
For dimension two, the results are comparable.
Surprisingly, when the dimension increases the RAM approach provides more
accurate results than the AM and the AMS algorithms.

One possible explanation is that in order to approximate the sample
covariance, the covariance adaptation in AM and ASWAM should be done using
the weight sequence $\eta_n=n^{-1}$ as this corresponds almost exactly to 
the usual sample covariance estimator. This setting was tried also; 
the results appear also in Table \ref{tab:normal-quantiles}. 
It seems that using such a sequence will indeed imply better 
results, when starting from $s_1\equiv I$
or $s_1 \equiv 10^{-4}\cdot I$. 
However, when the initial factor $s_1=10^{4}\cdot I$ was `too large',
this approach failed. This is probably due to the fact that in this 
case the eigenvalues of the covariance estimate 
can decay only slowly, at the speed $n^{-1}$.

Another explanation for the unsatisfactory performance of the AM and ASWAM
approaches is that in the experiments the adaptation was started right away,
not after a burn-in phase run with a fixed proposal covariance as
suggested in the original work \citep{saksman-am}. It is
expected that the AM and the ASWAM algorithms would perform better by a
suitable fixed proposal burn-in and perhaps with yet another step size
sequences. In any case, this experiment demonstrates one strength of the 
RAM adaptation 
mechanism, namely that it does not require such a burn-in period.

\begin{table*}
    \caption{Errors in Gaussian quantiles in different dimensions.
      The step sizes $\eta_n=n^{-1}$ were used for 
      covariance estimation for 
      AM\textsuperscript{\textdagger} and 
      ASWAM\textsuperscript{\textdagger}.}
    \label{tab:normal-quantiles} 
    \scriptsize
    \setlength{\tabcolsep}{0.8ex}
    \begin{tabular}{lccccccccccccccc}
        \toprule
        & \multicolumn{5}{c}{$s_1 \equiv I$} 
        & \multicolumn{5}{c}{$s_1 \equiv 10^{-4} \cdot I$} 
        & \multicolumn{5}{c}{$s_1 \equiv 10^{4} \cdot I$} \\
        \cmidrule(lr){2-6}
        \cmidrule(lr){7-11}
        \cmidrule(lr){12-16}
        $d$ &  2 & 4 & 8 & 16 & 32 
        &  2 & 4 & 8 & 16 & 32 
        &  2 & 4 & 8 & 16 & 32 \\
        \midrule
        AM & 0.21 & 0.33 & 1.25 & 6.83 & 33.87 
           & 0.20 & 0.33 & 1.26 & 6.79 & 35.73 
           & 0.21 & 0.33 & 1.24 & 6.83 & 32.49 \\
        ASWAM & 0.22 & 0.32 & 1.23 & 6.67 & 33.78 
           & 0.21 & 0.34 & 1.25 & 6.67 & 35.77 
           & 0.21 & 0.33 & 1.23 & 6.63 & 32.11 \\
        AM\textsuperscript{\textdagger} & 0.21 & 0.27 & 0.41 & 0.70 & 1.70 
           & 0.20 & 0.28 & 0.39 & 0.55 & 2.90 
           & 6.22 & 27.54 & 53.21 & 57.69 & 58.20 \\
ASWAM\textsuperscript{\textdagger} & 0.22 & 0.36 & 0.37 & 0.53 & 1.05 
           & 0.22 & 0.28 & 0.37 & 0.53 & 3.03 
           & 0.88 & 1.94 & 3.17 & 5.34 & 8.48 \\
        RAM & 0.21 & 0.27 & 0.37 & 0.52 & 1.03 
            & 0.22 & 0.27 & 0.38 & 0.62 & 2.51 
            & 0.22 & 0.28 & 0.45 & 0.75 & 1.61 \\
        \bottomrule
    \end{tabular} 
\end{table*}


\subsection{Mixture of separate Gaussians}
\label{sec:mixture} 

The last example concerns a mixture of two Gaussians distributions in $\R^d$
with mean vectors $m_1 \defeq [4,0,\ldots,0]^T$ and $m_2 \defeq -m_1$ and
with a common diagonal covariance matrix $\Sigma \defeq \diag(1,100,\ldots,100)$. In
such a case, the mixing will be especially problematic with respect to the
first coordinate.

\begin{table*}
\caption{Errors of the expectations of the first and the other
  coordinates in the mixture example.}
\label{tab:mixture} 
    \scriptsize
    \setlength{\tabcolsep}{0.8ex}
\begin{tabular}{lccccccccccc}
        \toprule
        & \multicolumn{5}{c}{$X^{(1)}$} 
        & \multicolumn{5}{c}{$X^{(2)},\ldots,X^{(d)}$} \\
        \cmidrule(lr){2-6}
        \cmidrule(lr){7-11}
        $d$ &  2 & 4 & 8 & 16 & 32 
            &  2 & 4 & 8 & 16 & 32 \\
        \midrule
AM & 0.04 & 0.05 & 0.08 & 1.69 & 3.87 & 0.08 & 0.11 & 0.15 & 0.19 & 0.27 \\
ASWAM & 0.04 & 0.06 & 0.10 & 1.82 & 3.86 & 0.08 & 0.11 & 0.14 & 0.18 & 0.27 \\
RAM & 0.07 & 0.21 & 0.66 & 1.34 & 1.77 & 0.05 & 0.08 & 0.11 & 0.16 & 0.29 \\
        \bottomrule
\end{tabular}
\end{table*}
Table \ref{tab:mixture} shows the root mean square error of the
expectation of the first coordinate $X^{(1)}$ and
the overall error for the rest $X^{(2)},\ldots,X^{(d)}$.
The errors in the first coordinate for the RAM are
significantly higher than for the AM and the ASWAM for dimensions 2, 4 and 8.
The estimates from all the algorithms are already quite 
unreliable in dimension $16$.
For the latter coordinates, the RAM approach seems to provide 
better estimates. Observe also that when comparing ASWAM with AM, the
results are also worse in the first coordinate and better in the rest, like
in the RAM approach. This indicates that the true optimal acceptance rate is
here probably slightly less than the enforced 23.4\%.

The example shows how the RAM approach finds the `local shape' of
the distribution. In fact, it is quite easy to see what happens if the
means of the mixture components would be made further and further
apart: there would be a stable point of the RAM algorithm that would
approach the common covariance of the mixture components. Such a
behaviour of the RAM approach is certainly a weakness in certain
settings, as this example, but it can be also advantageous. 
Notice also that even such a simple multimodal setting 
poses a challenge for the random walk based approaches.



\section{Discussion}
\label{sec:discussion} 

A new robust adaptive Metropolis (RAM) algorithm was presented. The
algorithm attains a given acceptance probability, and at the same time finds
an estimate of the shape of the target distribution. The algorithm can cope
with targets having arbitrarily heavy tails unlike the AM and ASWAM
algorithms based on the covariance estimate.
The RAM algorithm has some obvious limitations. It is not suitable for strongly
multi-modal targets, but this is the case for any random walk based approach.
For sufficiently regular targets, it seems to work well and
the experiments indicate that RAM is competitive with
the AM and ASWAM algorithms also in case of light-tailed targets having
second moments.

There are several interesting directions of further research that were not
covered in the present work. The RAM algorithm can be used also within
Gibbs sampling, that is, when updating a block of coordinate variables at a
time instead of the whole vector.  This approach is often very useful
especially when the target distribution $\pi$ consists of a product of
conditional densities, which is often the case with Bayesian hierarchical
models.  In such a setting, the computational cost of evaluating the ratio
$\pi(y)/\pi(x)$ after updating one coordinate block can be significantly less
than the full evaluation of $\pi(y)$. It would also be worth investigating 
the effect of different adaptation step sizes, perhaps 
even adaptive ones as suggested by
\cite{andrieu-thoms}. 

Regarding theoretical questions, the existence and uniqueness of the
fixed points of the approach could be verified in a more general
setting; the present work only covers elliptically symmetric and
product type target densities, which are too restrictive in practice. 
The experiments indicate the overall stability of the RAM algorithm;
see also Remark \ref{rem:stable-drift}. However, proving the stability
of RAM without prior bounds is directly related to the more general
open question on the stability of adaptive MCMC algorithms, or even
more generally to the stability of stochastic approximation. Having
the stability and more general conditions on the fixed points, one
could also prove the convergence of $S_n$ in a more general setting.


\section*{Acknowledgements} 

The author was supported by the Finnish Centre of Excellence in
Analysis and Dynamics Research and by the Finnish Graduate School in
Stochastics and Statistics. The author thanks Professor Christophe
Andrieu and Professor Eric Moulines for useful discussions on the
behaviour of the algorithm.



\appendix

\section{Regularity of directional mean acceptance probability}
\label{app:convolutions} 

\begin{proposition}
\label{prop:convolutions} 
Let $\pi$ and $q$ be probability densities on $\R^d$ and on $(0,\infty)$, 
respectively, and let $v\in\R^d$ be a unit vector.
The function $g:(0,\infty)\to(0,\infty)$ defined by
\[
    g(\theta) \defeq \int_0^\infty \int_{\R^d} 
  \min\left\{\pi(x),\pi(x+r\theta v)\right\} \ud x
  q(r) \ud r
\]
is continuous, $\lim_{\theta\to\infty} g(\theta) = 0$ and $\lim_{\theta\to
  0+} g(\theta) = 1$.
\end{proposition} 
\begin{proof} 
Suppose first that $\pi$ is a continuous probability density on
$\R^d$. Then, write
\[
    g(\theta)
    = \int_0^\infty \int_{A} 
  \min\left\{1,\frac{\pi(x+r\theta v)}{\pi(x)}\right\} \pi(x) \ud x
  q(r) \ud r
\]
where $A\defeq\{x\in\R^d:\pi(x)>0\}$ stands for the support of $\pi$.
Let $(\theta_n)_{n\ge 1}\subset(0,\infty)$ be any sequence and define
$f_{\theta}(x,r) \defeq \min\big\{1,\frac{\pi(x+r\theta v)}{\pi(x)}\big\}$.
Clearly, whenever $\theta_n$ converges to some $\theta$, then
$f_{\theta_n}(x,r)\to f_{\theta}(x,r)$ 
pointwise on $A\times(0,\infty)$ by the continuity of $\pi$. Since $|f_n(x,r)|\le
1$, the dominated convergence theorem yields that
$|g(\theta_n)-g(\theta)|\to 0$, establishing the continuity.
For any sequence $\theta_n\to 0+$ one clearly has $f_{\theta_n}(x,r)\to 1$,
and for any sequence $\theta_n\to\infty$ one obtains 
$f_{\theta_n}(x,r)\to 0$, establishing the claim.

Let us then proceed to the general case. Let $\epsilon>0$ be arbitrary.
We shall show that there exists a continuous 
probability density $\tilde{\pi}$ 
on $\R^d$ such that 
\[
    \int_{\R^d} |\tilde{\pi}(x) - \pi(x) | \ud x <\epsilon.
\]
Having such $\tilde{\pi}$, one can bound the difference 
\begin{equation*}
    \left|g(\theta) - 
    \int_0^\infty \int_{A} 
  \min\left\{1,\frac{\hat{\pi}(x+r\theta v)}{\hat{\pi}(x)}\right\} \hat{\pi}(x) \ud x
  q(r) \ud r\right|
  \le \int_{\R^d} |\pi(x) - \tilde{\pi}(x)| \ud x <\epsilon
\end{equation*}
establishing the claim.

Let us finally verify that such a continuous probability density 
$\tilde{\pi}$ exists. Approximate $\pi$ first by smooth non-negative 
functions $\pi_n$ such that $\int_{\R^d} |\pi(x) - \pi_n(x)| \ud
x\to 0$, and then normalise them to probability densities
$\tilde{\pi}_n(x) \defeq c_n \pi_n(x)$.
Clearly, the constants $c_n \defeq (\int_{\R^d} \pi_n(z) \ud z)^{-1}\to
1$, and so $\int_{\R^d} |\pi(x) - \tilde{\pi}_n(x)| \ud x
\le \int_{\R^d} |\pi(x) - \pi_n(x)| \ud x + |1-c_n|\to 0$.
\end{proof} 


\section{Proofs of convergence}
\label{sec:slln-proof} 

\begin{proof}[Theorem \ref{th:slln}] 
Let $0<a\le b<\infty$ be arbitrary constants
and denote by $\spc{S}_{a,b}\subset\R^{d\times d}$ the set of
all lower triangular matrices with positive diagonal, such that the
eigenvalues of $s s^T$ are within $[a,b]$.
Let $P_s$ stand for the random walk Metropolis kernel with a proposal
density $q_s(z) \defeq \det(s)^{-1} q(s^{-1} z)$, that is,
for any $x\in\R^d$ and any Borel set $A\subset\R^d$
\begin{align*}
    P_s(x,A) \defeq \charfun{A}(x)
    &\left(1-\int_{\R^d}
        \min\left\{1,\frac{\pi(y)}{\pi(x)}
          \right\}
        q_s(y-x)\ud y\right)\\
    &+ \int_A \min\left\{1,\frac{\pi(y)}{\pi(x)}
          \right\} q_s(y-x) \ud y.
\end{align*}
We shall use the notation $P_s f(x) \defeq \int_{\R^d} f(y) P_s(x,\ud y)$
to denote the integration of a function with respect to the kernel
$P_s$.

Let us check that the following assumptions are satisfied.
\begin{enumerate}[({A}1)]
  \item \label{a:invar} 
    For all possible $s\in\spc{S}_{a,b}$, the kernels $P_s$ have a unique
    invariant probability distribution $\pi$ for which 
    $\int_{\R^d} P(x,A) \pi(\ud x) = \pi(A)$ for any Borel set
    $A\subset\R^d$.
  \item \label{a:drift} 
    There exist a Borel set $C\subset\R^d$,
    a function $V:\R^d\to[1,\infty)$, 
    constants $\delta,\lambda\in(0,1)$ and $b<\infty$,
    and a probability measure $\nu$ concentrated on $C$
    such that
  \begin{eqnarray*}
    P_s V(x) &\le & \lambda V(x) + \charfun{C}(x)b 
    \qquad\text{and}
    \\
    P_s(x,A) &\ge& \charfun{C}(x) \delta \nu(A)
  \end{eqnarray*}
  for all possible $x\in\R^d$, $s\in\spc{S}_{a,b}$ and all Borel sets $A\subset\R^d$.
  \item \label{a:kernel-lip} 
  For all $n\ge 1$ and any $r\in(0,1]$, there is a constant 
  $c'=c'(r)\ge 1$ such that
  for all $s,s'\in\spc{S}_{a,b}$,
  \begin{equation*}
      \sup_{x\in\R^d} \frac{|P_s f(x) - P_{s'} f(x) |}{V^r(x)}
      \le c' | s - s'| \sup_{x\in\R^d} \frac{|f(x)|}{V^r(x)}.
  \end{equation*}
  \item \label{a:adapt-bound} 
        There is a constant $c<\infty$ such that 
        for all $n\ge 1$,  $s\in \spc{S}_{a,b}$, $x\in\R^d$ and $u\in\R^{d}$
        the bound $|H(s,x,u)| \le c$ holds. 
\end{enumerate}
The uniqueness of the invariant distribution (A\ref{a:invar}) follows
by observing that the kernels $P_s$ are irreducible, aperiodic and
reversible with respect to $\pi$ \citep[see, e.g.][]{nummelin-mcmcs}. 
The simultaneous drift and minorisation condition 
(A\ref{a:drift}) and the continuity condition 
was established by \cite{andrieu-moulines}.
The continuity condition (A\ref{a:kernel-lip})
was established by \cite{andrieu-moulines} for Gaussian proposal
distributions and was extended to cover the Student proposal in
\citep{vihola-asm}. The bound (A\ref{a:adapt-bound}) is easy to verify.

Assumption \ref{a:stability} ensures that for any $\epsilon>0$ there exist
constants $0<a_\epsilon\le b_\epsilon<\infty$ such that all the eigenvalues
of $S_n S_n^T$ stay within the interval $[a_\epsilon,b_\epsilon]$ at least
with probability $\P(\Omega_0)-\epsilon$. This is enough to ensure that the strong law of
large numbers holds by \citep[Proposition 6]{andrieu-moulines}. For
details, see also \citep[Theorem 2]{saksman-vihola} and 
\citep[Theorem 20]{vihola-asm}.
\end{proof} 

\begin{proof}[Theorem \ref{th:convergence}]
The proof follows by \citep[Theorem 5]{andrieu-moulines-volkov} by
using a similar technique as in the proof of Theorem \ref{th:slln}.
Consider the Lyapunov function $w_\pi(R)$ defined in 
Theorem \ref{th:lyapunov}.
It is straightforward to
verify items 1--4 of \citep[Condition 1]{andrieu-moulines-volkov}
when we take $\Theta$ to be the space of symmetric positive definite
matrices and consider $S_nS_n^T\in\Theta$.
The compact sets are of the form 
$K=\{s s^T:s\in\spc{S}_{a_\epsilon,b_\epsilon}\}$ with
$a_\epsilon,b_\epsilon$ as in the proof of 
Theorem \ref{th:slln}.
Item 5 follows by invoking \citep[Proposition
6]{saksman-vihola} with $f_{\theta}\big((x,u)\big) = H(\theta,x,u)$.
\end{proof}


\end{document}